\documentclass[10pt]{article}
\usepackage{dsfont}             %
\usepackage{amsmath}
\usepackage{amsthm}
\usepackage{amssymb}
\usepackage{algorithmicx}
\usepackage{algorithm}
\usepackage{algpseudocode}
\usepackage{graphicx}
\usepackage{hyperref}
\usepackage{mathabx}
\usepackage[letterpaper,margin=1in]{geometry}
\renewcommand{\setminus}{-}
\newcommand{\R}{\mathbb{R}}    
\theoremstyle{plain}
\newtheorem{theorem}{Theorem}[section]
\newtheorem{lemma}[theorem]{Lemma}
\newtheorem{definition}[theorem]{Definition}

\newtheorem{prop}[theorem]{Proposition}

\newtheorem{fact}{Fact}

\theoremstyle{remark}

\newtheorem*{note}{Note}

\newcommand{\eps}{\varepsilon}
\newcommand{\opt}{\text{OPT}}

\newcommand{\dist}{\text{dist}}

\newcommand{\cost}{\text{cost}}

\newcommand{\event}{\mathcal{E}}

\newcommand{\calD}{\mathcal{D}}

\newcommand{\calL}{\mathcal{L}}

\newcommand{\calM}{\mathcal{M}}

\newcommand{\calP}{\mathcal{P}}

\newcommand{\tw}{\tilde{w}}

\newcommand{\local}{L}

\newcommand{\globalS}{\mathcal{G}}

\newcommand{\Vloc}{V_{L}}
\newcommand{\Vglob}{V_{\opt}}

\newcommand{\improv}{\text{Improv}}

\newcommand{\irr}{\text{moat}}

\def\eg{\textit{e.g., }}
\def\ie{\textit{i.e., }}

\begin{document}

%
%

\title{A Fast Approximation Scheme for Low-Dimensional $k$-Means}
\author{Vincent Cohen-Addad\\ University of Copenhagen}
\date{}

\maketitle
\begin{abstract}
  We consider the popular $k$-means problem in $d$-dimensional
  Euclidean space. Recently Friggstad, Rezapour, Salavatipour [FOCS'16]
  and Cohen-Addad, Klein, Mathieu
  [FOCS'16] showed that the standard local search algorithm yields a 
  $(1+\eps)$-approximation
  in time $(n \cdot k)^{1/\eps^{O(d)}}$, giving the first 
  polynomial-time approximation
  scheme for the problem in low-dimensional Euclidean space.
  While local search achieves optimal approximation guarantees, it is not 
  competitive with the state-of-the-art heuristics such as the famous 
  $k$-means++ and $D^2$-sampling algorithms.

  In this paper, we aim at bridging the gap between theory and practice by
  giving a $(1+\eps)$-approximation algorithm for low-dimensional 
  $k$-means 
  running in time $n \cdot k \cdot (\log n)^{(d\eps^{-1})^{O(d)}}$,
  and so  matching the running time of the $k$-means++ and 
  $D^2$-sampling heuristics 
  up to polylogarithmic factors.
  We speed-up
  the local search approach by making a non-standard use
  of randomized dissections that allows to 
  find the best local move efficiently using a quite simple
  dynamic
  program. We hope that our techniques could help
  design better
  local search heuristics for geometric problems.
\end{abstract}

\section{Introduction}
The $k$-means objective is arguably the most popular clustering 
objective among practitioners. While originally motivated by applications
in image compression, the $k$-means problem has proven to be a 
successful objective to optimize in order to pre-process and extract
information from datasets. 
Its most successful applications are now stemming from 
machine learning problems like for example learning 
mixture of Gaussians, Bregman clustering, or DP-means~\cite{Feldman:2986698,LucicBK16,Bachem.3045142}. 
Thus, it has become a classic problem in both machine learning
and theoretical computer science.

Given a set of points in a metric space, the $k$-means
problem asks for a
 set of $k$ points, called \emph{centers},
that minimizes the sum of the squares of the 
distances of the points to their closest center.

The most famous algorithm for $k$-means is arguably the 
Lloyd\footnote{Also referred to as Lloyd-Forgy} heuristic introduced
in the 80s~\cite{Llo82} and sometimes 
referred to as ``the $k$-means algorithm''.
While this algorithm is very competitive in
practice and yields empirically 
good approximate solutions to real-world inputs,
it is known that its running time can be exponential in the input 
size and that it can return arbitrarily bad solutions
in the worst-case (see~\cite{ArV06}). This induces a gap between
theory and practice.

Thus, to fix this unsatisfactory situation,
Arthur and Vassilvitskii~\cite{ArV07} have designed a variant
of the  Lloyd heuristic, called 
$k$-means++,
and proved that it achieves an $O(\log k)$ approximation.
The $k$-means++ algorithm has now become 
a standard routine that is part of several machine learning libraries
(\eg\cite{Petal11}) and is widely-used in practice.
While this has been a major step for reducing the gap between
theory and practice, it has remained an important problem as to
conceive algorithms with nearly-optimal approximation guarantee.

Unfortunately, the $k$-means problem is known to be APX-Hard
even for (high dimensional) Euclidean inputs~\cite{ACKS15}. 
Hence, to design
competitive approximation schemes it is needed to restrict our
attention to classes of ``more structured'' inputs that are important 
in practice.
The low-dimensional Euclidean inputs form a class of 
inputs that naturally arise in image processing and
machine learning (see examples of~\cite{Petal11} or
in~\cite{KMNPSW04}).
Thus, finding a polynomial
time approximation scheme (PTAS) for 
$O(1)$-dimensional
Euclidean inputs of $k$-means  has been
an important research problem for the last 20 years 
since the seminal work of~\cite{IKI94}. 
Recently, Friggstad et
al.~\cite{FRS16a}
and Cohen-Addad et al.~\cite{CAKM16} both showed that the classic
local search heuristic with neighborhood of magnitude $(d/\eps)^{O(d)}$
achieves a $(1+\eps)$-approximation.
While this has been an important result for the theory community, it 
has a much weaker impact for practitioners since the running time
of the algorithm is $n^{(d/\eps)^{O(d)}}$. 
Therefore, to reduce the gap between
theory and practice, it is natural to ask for
near-optimal approximation algorithms with competitive
running time.
This is the goal of this paper.

\paragraph{Fast local search techniques are important.}
The result of Friggstad et al. and Cohen-Addad et al. has been 
preceded by several result showing that local search achieves good
approximation bound or even exact algorithms for various problems 
(see~\eg\cite{nabil1,Krohn14,BaV15,MarxP15,Chan09}).
Furthermore, there is a close relationship
between local search and clustering showing that the standard local
search heuristics achieve very good approximation guarantees
in various settings 
(in addition to the two 
aforementioned papers, 
see~\cite{AGKMMP04,KMNPSW04,Cohen-AddadS17,CM15,BaV15,MakarychevMSW16}).
Moreover local search approaches are extremely popular in practice 
because they are easy to tune, easy to implement, and easy
to run in parallel. 

Thus, it has become part of the research agenda of the theory community
to develop fast local search approaches while preserving the 
guarantees on the quality of the output (see \eg\cite{Cygan13,nabil2}).

\subsection{Our Results}
We show that our fast local search algorithm 
(Algorithm~\ref{algo:Clustering})
yields a PTAS for the slightly more
general variants of the $k$-means problem where centers can 
have an opening cost
(a.k.a. weight).

\begin{theorem}
  \label{thm:main}
  There exists a randomized
  algorithm (Algorithm~\ref{algo:Clustering}) 
  that returns a
  $(1+\eps)$ approximation to the center-weighted $d$-dimensional 
  Euclidean
  $k$-means problem in time $n \cdot k \cdot 
  (\log n)^{(d \eps^{-1})^{O(d)}}$ with probability at least $1/2$.
\end{theorem}

We would like to remark that the doubly exponential dependency 
in $d$ is needed: 
Awasthi et al.~\cite{ACKS15} showed that the Euclidean $k$-means 
is APX-Hard when $d = \Omega(\log n)$.
Note that it is possible to obtain an arbitrarily small
probability of failure $p > 0$ by repeating the algorithm $\log 1/p$ 
times.

As far as we know, this is the first occurrence of a local search 
algorithm whose neighborhood size only impacts the running time 
by polylogarithmic factors.

\subsection{Other Related Work}
The $k$-means problem is known to be NP-Hard, even when restricted
to inputs lying in the Euclidean plane 
(Mahajan et al. ~\cite{MNV12}, and Dasgupta and Freud~\cite{DaF09}) and
was recently shown to be APX-hard in Euclidean space of dimension
$\Omega(\log n)$ (\cite{ACKS15}).

There has been a large body of work on approximation algorithms 
for the Euclidean $k$-means problem (see \eg\cite{KMNPSW04}),
very recently
Ahmadian et al.~\cite{ahmadian2016better} gave a 6.357-approximation
improving over the 9-approximation of Kanungo et al.~\cite{KMNPSW04}.


Given the hardness results, researchers have focused on 
different scenarios. There have been various 
$(1+\varepsilon)$-approximation algorithms when $k$ is 
considered a fixed-parameter
(see \eg\cite{FeL11,KSS10}). Another successful approach
has been through the definition of ``stable instances'' to characterize
the real-world instances stemming from machine learning and data 
analysis 
(see for example 
\cite{AwS12,ABS10,BaL16,BBG09,BiL12,KuK10,ORSS12,Cohen-AddadS17,AngelidakisMM17}) 
or in the context of smoothed analysis (see for 
example \cite{ArV09,AMR11}). 
In the case of low-dimensional inputs, Bandyapadhyay and
Varadarajan showed that local search with neighborhood of size
$\eps^{-O(d)}$ achieves a $(1+\eps)$-approximation~\cite{BaV15} when
allowed to open $O(\eps k)$ extra centers. As mentioned before,
this results has been improved by Friggstad et al.~\cite{FRS16a}
and Cohen-Addad et al.~\cite{CAKM16} who showed that even when
constrained to open exactly $k$ centers, local search achieves
a $(1+\eps)$-approximation.

\paragraph{Related work on local search}
Local Search heuristics belong to the toolbox of all practitioners,
(see Aarts and Lenstra~\cite{AaL97} for a general introduction).
As mentioned before, there is a tight connection between local search
and clustering:
Arya et al.~\cite{AGKMMP04} proved that local search with a neighborhood size of $1/\varepsilon$ yields a $3+2\varepsilon$ approximation to $k$-median. For the $k$-means problem,
Kanungo et al.~\cite{KMNPSW04} showed a similar result by proving 
that the approximation guarantee of local search for Euclidean $k$-means
is $9+\varepsilon$.
For more applied examples of local search an clustering 
see~\cite{BBLM14,BeT10,GMMMO03,tuzun1999,Ghosh2003150,Ardjmand201432,hansen2001variable}.
For other theoretical example of local search for clustering, 
we refer to~\cite{SAD04,DGK02,FrZ16,HaM01,YSZUC08}.

\paragraph{Related work on $k$-median.}
The $k$-median problem has been widely studied.
For the best known results in terms of approximation for general
metric spaces inputs we refer to Li and Svensson~\cite{LiS13} and 
Byrka et al.~\cite{BPRST15}. More related to our results are the 
approximation schemes for $k$-median in low-dimensional Euclidean
space given by Arora et al.~\cite{ARR98} who gave  a
$(1+\eps)$-approximation algorithm running in time
$n^{\eps^{-O(d)}}$. This was later improved by Kolliopoulos and 
Rao~\cite{KoR07} who obtained a running time of $2^{\eps^{-O(d)}} n\cdot
\text{polylog}~n$. Quite surprisingly, it is pretty unclear whether 
the techniques used by Arora et al. and Kolliopoulos and Rao could be
used to obtain a $(1+\eps)$-approximation for the $k$-means problem; 
this has induced a 20-year gap between the first PTAS for $k$-median
and the first PTAS for $k$-means in low-dimensional Euclidean space.
See Section~\ref{sec:techniques} for more details.

\subsection{Overview of the Algorithm and the Techniques}
\label{sec:techniques}
Our proof is rather simple.
Given a solution $L$, our goal is to identify -- in near-linear
time -- a minimum
cost solution $L'$
such that $|L - L'| + |L' - L| \le \delta$ for some (constant) parameter 
$\delta$. If $\cost(L) - \cost(L') = O(\eps \cost(\opt)/k)$,
then we can immediately apply the result of  
Friggstad et al.~\cite{FRS16a}
or Cohen-Addad et al.~\cite{CAKM16}: the solution $L$ is 
locally optimal and its cost is at most $(1+O(\eps))\cost(\opt)$.
Finding $L$ has
to be done in near-linear time since we could repeat this 
process up to $\Theta(k)$ times
until reaching a local optimum.
Hence, the crux of the algorithm is to efficiently identify $L'$,
it proceeds as follows (see Algorithm~\ref{algo:Clustering} for a full description): 
\begin{enumerate}
\item Compute a random recursive decomposition of $L$
  (see Section~\ref{sec:diss}); 
\item Apply dynamic programming on the recursive decomposition;
  we show that there exists a near-optimal solution
  whose interface between different regions
  has small complexity.
\end{enumerate}

To obtain our recursive decomposition we make a quite non-standard
use of the classic quadtree dissection techniques 
(see Section~\ref{sec:diss}).
Indeed, the $k$-means problem is famous for being ``resilient'' to 
random quadtree approaches -- this is partly why a PTAS for the 
low-dimensional $k$-median problem was obtained 20 years ago but the 
first PTAS for the $k$-means was only found last year. More precisely,
the classic quadtree approach (which works well for $k$-median) 
defines portals on the boundary of 
the regions of the dissection and forces the clients of a given region
that are served by a center that is in a different region 
(in the optimal solution) to make a detour through the closest portal.
This is a key property as it allows to bound the complexity of 
$\opt$ between regions. Unfortunately, when dealing with squared 
distances, making a detour could result in a dramatic cost increase
and it is not clear that it could be compensated by the fact that 
the event of separating a client form its center happens with small
probability (applying the analysis of 
Arora et al.~\cite{ARR98}
or Kolliopoulos et al.~\cite{KoR07}). This problem comes from
the fact that some facilities of $\opt$ and $L$ might be too
close from the boundary of the dissection (and so, their
clients might have to make too important detour (relative
to their cost in $\opt$) through
the portals). We call these facility the ``moat'' facilities
(as they fall in a bounded-size ``moat'' around the boundaries).

We overcome this barrier by defining a more structured
near-optimal solution as follows: (1) when a facility of $\opt$ is too close from the boundary of our dissection we simply remove it and (2) if a facility of the current solution $L$ is 
too close from a boundary of our dissection we add it to $\opt$.

%

Of course, this induces two problems: first we have to bound
the cost of removing the facilities of $\opt$ and second we
have to show that adding the facilities of $L$ does not result
in a solution containing more than $k$ centers.
This is done through some technical lemmas and using the concept
of \emph{isolated facilities} inherited from
Cohen-Addad et al.~\cite{CAKM16}. Section~\ref{sec:structthm}
shows the existence of near-optimal solution $S^*$ whose set
of $\irr$ facilities (\ie facilities that are too close from
the boundaries) is exactly the set of $\irr$ facilities
of $L$ (and so, we already know their location and so the
exact cost of assigning a given client to such a facility).

We now aim at using the result of Friggstad et al.~\cite{FRS16a}
and Cohen-Addad et al.~\cite{CAKM16}: if for any sets
$\Delta_1 \subseteq L$ and $\Delta_2 \subseteq S^*$,
we have $\cost(L) - \cost(L - \Delta_1 \cup \Delta_2) =
O(\eps \cost(\opt)/k)$, then we have $\cost(L) \le
(1+O(\eps)\cost(S^*)$.
Thus, we provide a dynamic program (in Section~\ref{sec:DP}) 
that allows to find the best solution $S$ that is such that
(1) its set of $\irr$
facilities is exactly the set of $\irr$ facilities
of $L$ and (2) $|S-L|+|L-S| \le \delta$ for some fixed
constant $\delta$. The dynamic program simply ``guess'' the
approximate location of the centers of $(L-S) \cup (S-L)$,
we show that since each such center is far from the boundary,
its location can be approximated.

\subsection{Preliminaries}
In this article, we consider the $k$-means problem in a $d$-dimensional Euclidean space:
Given a set $A$ of points (also referred to as clients) and candidate centers $C$ in $\R^d$,
the goal is to output a set $S\subseteq C$ of size $k$ that minimizes:
$\sum_{a \in A} \dist(a,C)^2$, where $\dist(a,C) = \min_{c \in C} \dist(a,c)$. We refer to $S$ as a set of centers or facilities.
Our results naturally extends to any objective function of the form 
$\sum_{a \in A} \dist(a,C)^p$ for constant $p$. For ease of exposition, we focus on the 
$k$-means problem. As Friggstad et al~\cite{FRS16a}, we also 
consider the more general version
called weighted $k$-means for which, in addition of the sets $A$ and $C$, we are given a weight
function $w : C \mapsto \R_+$ and the goal is to minimize
$\sum_{c \in C} w(c) + \sum_{a \in A} \dist(a,C)^2$.

A classic result of Matousek~\cite{Mat00} shows that, if $C = \R^d$, it is possible to compute in linear
time a set $C'$ of linear size (and polynomial dependency in $d$ and $\eps$) such that the optimal solution using the 
centers in $C'$ cost at most $(1+\eps)$ times the cost of the optimal solution using
$C$. Hence, we assume without loss of generality that 
$|C|$ has size linear in $|A|$ and we
let $n = |A| + |C|$.

\paragraph{Isolated Facilities.} We make use of the notion of isolated facilities
introduced by Cohen-Addad et al.~\cite{CAKM16} defined as follows.
Let $\eps_0 <1/2$ be a positive constant and $\local$ and $\globalS$ be two solutions for the 
Euclidean $k$-means problem.

\begin{definition} 
  Let $\eps_0 <1/2$ be a positive number and let $\local$ and $\globalS$ be two solutions for the $k$-clustering problem with parameter $p$.
  Given a facility $f_0 \in \globalS$ and a facility $\ell\in \local$, 
we say that the pair
  $(f_0,\ell)$ is  {\em 1-1 $\eps_0$-isolated} if 
  most of the clients served by $\ell$ in $\local$ are served by 
  $f_0$ in $\globalS$, and
  most of the clients served by $f_0$ in $\globalS$ are served by 
  $\ell$ in $\local$; formally, if
  $$|\Vloc(\ell) \cap \Vglob(f_0)  |   \ge \max \left\{ \begin{array}{l} 
                                                     (1-\eps_0)   
                                                     |\Vloc(\ell)|,    \\
   (1-\eps_0) |\Vglob(f_0)| \end{array}\right\} $$
\end{definition}
When $\eps_0$ is clear from the context we refer to 1-1 $\eps_0$-isolated 
pairs as 1-1 isolated pairs.

\begin{theorem}[Theorem III.7 in~\cite{CAKM16}]\label{thm:deletion}
  Let $\eps_0 <1/2$ be a positive number and let $\local$ and $\globalS$ be two solutions for the $k$-clustering problem with exponent $p$.
  Let $\bar k$ denote the number of facilities $f$ of $\globalS$  that
  are not in a 1-1 $\eps_0$-isolated region. 
  There exists a constant $c$ and a set $S_0$ of facilities of $\globalS$ of size at least $\eps_0^3 \bar k/6$ that can be removed 
  from $\globalS$ at low cost:  $\cost(\globalS \setminus S_0) \le (1+c \cdot\eps_0) \cost(\globalS) +  c \cdot \eps_0\,\cost(\local)$.
\end{theorem}

Observe  that since $\eps_0 < 1/2$, each facility of $\local$ belongs to at most one isolated region.
Let $\tilde \globalS$ denote the facilities of $\globalS$ that are not
in an isolated region.  
In the rest of the paper, we will use it with $\eps_0 = \eps^3$.

\subsection{Fast Local Search}
This section is dedicated to our the description of the fast local
search algorithm for the $k$-means problem.
It relies on a dynamic program called
\textsc{FindImprovement} that allows to find the best improvement of
the current solution in time $n \cdot \text{poly}_{\eps,d}(\log n)$.
We then show that total number of iterations of the do-while loop is 
$O(k)$.

\algdef{SE}[DOWHILE]{Do}{doWhile}{\algorithmicdo}[1]{\algorithmicwhile\ #1}%
\algblockdefx{MRepeat}{EndRepeat}{\textbf{repeat}}{}
\algnotext{EndRepeat}

\begin{algorithm}
  \caption{Fast Local Search for $k$-Means}
  \label{algo:Clustering}
  \begin{algorithmic}[1]
    \State \textbf{Input:} An $n$-element client set $A$, an $m$-element 
    candidate center set $C$, a positive integer 
    parameter $k$, an opening cost function 
    $w : C \mapsto \R_+$, and an
    error parameter $0< \eps < 1/2$, 
    \State\label{line:preproc}
    $L \gets$ $O(1)$-approximation.
    \State Round up the weights of the candidate centers
    to the closest 
    $(1+\eps)^i \cdot \eps \cdot \cost(L)/n$ 
    for some integer $i$.
    \Do \label{line:do}
    \State $\improv \gets 0$
    \State $L^* \gets L$
    \MRepeat{~~\textbf{log \textit{k} times} 
      (to boost the success probability)}
    \State Compute a random decomposition $\calD$ of $L$ (as in
    Sec.~\ref{sec:diss}). 
    Let $\calM$ be the $\irr$ centers of $L$.
    \State $L' \gets$ output of 
    \textsc{FindImprovement}($L$, $\calD$, $\calM$, $d^{O(d)}\eps^{-O(d)}$)
    \If{$\improv \le \cost(L) - \cost(L')$}
    \State $\improv \gets \cost(L) - \cost(L')$
    \State $L^* \gets L'$
    \EndIf
    \EndRepeat
    \State $L \gets L^*$
    \doWhile{$\improv > \eps \cost(L)/k$}\label{line:while}
    \State \textbf{Output:} $L$     
  \end{algorithmic}
\end{algorithm}

\section{Dissection Procedure}
\label{sec:diss}
In this section, we recall the classic definition of 
quadtree dissection.
For simplicity we give the definition for $\R_2$, 
the definition directly generalizes to any fixed dimension $d$, 
see Arora~\cite{Arora97} and Arora et al.~\cite{ARR98} for a 
complete description.
Our definition of quadtree is standard and follows the 
definition of~\cite{Arora97,ARR98}, our contribution in the 
structural properties we extract from the dissection and 
is summarized by Lemma~\ref{lemma:keyprob}.

Let $\calL$ be the length of the bounding box of the client
set $A$ (\ie
the smallest square containing all the points in $A$).  
Applying standard preprocessing techniques, 
see in~\cite{KoR07} and~\cite{ARR98}, 
it is possible
to assume that the points lie on a unit grid of size polynomial in the 
number of input points. This incurs an addititive error of $O(\opt/n^c)$ 
for
some constant $c$ and yields $\calL = n \cdot \text{poly}(\eps^{-d})$.

We define a \emph{quadtree dissection} $\calD$ of a set of 
points $\calP$ as follows.
A dissection of (the bounding box of) $\calL$ is a recursive 
partitioning into smaller squares. We view
it as a $4$-ary tree whose root is the bounding
box of the input points. Each square in the tree is partitioned into
$4$ equal squares, which are its children. We
stop partitioning a square if it has size $< 1$ (and
therefore at most one input point). It follows that the
depth of the tree is $\log \calL = O(\log n)$. 
Standard techniques show that such a quadtree can be computed 
in $n\cdot \log n \cdot \text{poly}(\eps^{-d})$, 
see~\cite{KoR07} for more details. The total number
of nodes of the quadtree is 
$n\cdot \log n \cdot \text{poly}(\eps^{-d})$.

Given two integers $a,b \in [0,\calL)$, the $(a,b)$-shifted
dissection consists in shifting the $x$- and $y$- coordinates
of all the vertical and horizontal lines by $a \mod \calL$ and
$b \mod \calL$ respectively. For a shifted dissection, we naturally
define the level of a bounding box to be its depth in the quadtree.
From this, we define the level of a line to be the level
of the square it bounds.

For a given square of the decomposition, each boundary of 
the square defines a \emph{subline} of one of the $2\calL$ lines
of the grid. It follows that each line at level $i$ consists
of $2^i$ sublines of length $\calL/2^i$.

Given an $(a,b)$-shifted quadtree dissection of a set of $n$ points, 
and given a set of points $U$,  
we say that a point $p$ of $U$ is an \emph{$i,\gamma$-$\irr$}
point
if it is at distance less than $\gamma \cdot \calL/2^{i}$ of a line
of the dissection that is at level $i$.
We say that a point $p$ of $U$ is a \emph{$\gamma$-$\irr$}
point
if there exists an $i$ such that $p$ is a $i,\gamma$-$\irr$
point. When $\gamma$ is clear from the context, we simply
call such a point a $\irr$ point.
We have:
\begin{lemma}
  \label{lemma:prob}
  For any $p \in U$, the probability that 
  $p$ is a $\gamma$-$\irr$ point
  is at most $\gamma\log \calL = O(\gamma \log n)$.
\end{lemma}
\begin{proof}
  Let $i$ be an integer in $[0,\ldots,\log \calL]$ 
  and consider the horizontal lines at level $i$
  (an analogous 
  reasonning applies to the vertical lines).
  By definition, the number of dissection
  lines that are at distance 
  at most $\gamma \calL/2^i$ 
  from $p$ is
  $\gamma \calL/2^i$.
  We now bound the probability that one of them is 
  at level $i$ (and so the probability of $p$ being
  at distance less than $\gamma \calL/2^i$ of a horizontal
  line of length $\calL/2^i$). 
  For any such line $l$,
  we have:
  $Pr_a[l \text{ is at level }i] = 2^i/\calL$.
  Hence, $Pr_a[p$ is a $i,\gamma$-$\irr$  point$]
  \le \sum_{l: \dist(l,p) \le  \gamma \calL/2^i} 
  Pr_a[l \text{ is at level }i] \le \gamma$.
  The lemma follows by taking a union bound over all $i$.
\end{proof}


We now consider an optimal solution $\opt$ and
any solution $L$. In the following, we will focus
on $\gamma$-$\irr$ centers of $L$ and $\opt$ for
$\gamma = \eps^{13}/\log n$. In the rest of the
paper, $\gamma$ is fixed to that value and so
$\gamma$-$\irr$
centers are simply called $\irr$ centers.

Let $\iota$ the facilities of $\opt$ that are not
1-1 isolated.
We define a weigth function $\tilde{w} : 
L \cup \opt \cup \opt \mapsto \R_+$ as follows.
For each facility $s \in \iota$ 
we define $\tw(s)$ as the sum of $w(s)$ and the cost
of serving all the clients served by $s$ in $\opt$
by the closest facility $\ell$ in $L$ plus $w(\ell)$.
For each facility $s \in L$ 
we define $\tw(s) = w(s)
+ \sum_{c~\text{served by $s$ in $L$}} \dist(c,s)$.
Similarly, for each $s \in \opt - \iota$,
we let $\tw(s) = w(s)
+ \sum_{c~\text{served by $s$ in $\opt$}} \dist(c,s)$.


We show:
\begin{lemma}
  \label{lem:newreass}
  There exists a constant $c_0$ such that
  $\sum_{s \in \iota} \tw(s)
  \le  c_0 (\cost(\opt) + \cost(L))$.  
\end{lemma}
\begin{proof}
  Consider a facility $s \in \iota$
  and the closest facility $l$ in $L$ that (1)
  serves in $L$ at least one client that is served by $s$ in $\opt$
  and that (2) minimizes the
  following quantity:
  $\eta = \min_c \dist(c, l)^2 + \dist(c, s)^2$. Let $c^*$ be a client
  that minimizes the quantity $\dist(c,l)^2+\dist(c,s)^2$.
  Let $N(s)$ be the set of all clients
  served by $s$ in $\opt$. We have:
  $|N(s)| \eta \le \sum_{c \in N(s)} \dist(c,L)^2+\dist(c,\opt)^2$.
  
  We have that the total cost of sending all the clients in $N(s)$
  is at most (by triangle inequality):
  $\sum_{c \in N(s)} (\dist(c, s) + \dist(s,l))^2 \le
  \sum_{c \in N(s)} (\dist(c, s) + \dist(s,c) + \dist(c,l))^2$.
  Note that there exists a constant $c_0$ such that the above sum is at most
  $c_0 \sum_{c \in N(s)}\dist(c, s)^2 + \dist(s,c^*)^2 + \dist(c^*,l)^2$ and so at
  most
  $c_0 (|N(s)|\eta + \sum_{c \in N(s)}\dist(c, \opt)^2)$. The lemma follows
  by combining with the
  above bound on $|N(s)|  \eta$.
\end{proof}

In the following we denote by $\phi : \iota \mapsto L$
the mapping from
each non-isolated facility of $\opt$ to its
closest facility in $L$. We define $\iota_L$ to
be the set of non-isolated facilities of $L$.

We define Event $\event(L \cup \opt, \tilde{w})$
as follows:
\begin{enumerate}
\item The set of $\irr$
  centers $S_1$ of $L \cup \iota$ is such that
  $\tw(S_1) = \sum_{c \in S_1} \tw(c) 
  \le \eps^9 \sum_{c \in L \cup \iota} \tw(c)
  = \eps^9 \tw(L \cup \iota)$, and
\item The set of $\irr$
  centers $S_1$ of $\iota_L \cup \iota$ is such that 
  $|S_0| \le \eps^9 |\iota_L \cup \iota| = \eps^9\bar{k}$
  (recall that $\bar{k}$ is the number of non 1-1-isolated
  facilities of $\opt$ (and so of $L$ as well)).
\end{enumerate}
The following lemma follows from Lemma~\ref{lemma:prob},
applying Markov's
inequality and taking a union bound over the probability of failures of 
property (1) and (2).

\begin{lemma}
  \label{lemma:keyprob}
  The probability that Event $\event(L \cup \opt, \tilde{w})$
  happens is at least $1/2$.
\end{lemma}
\begin{proof}
  We apply Lemma~\ref{lemma:prob} and obtain that any
  element of $L \cup \opt$
  is a $\irr$ center
  with probability $O(\rho^{-1} \log n)$. Since 
  $\rho^{-1} =  (c \eps^{-12} \log n)^{-1}$,
  we obtain that this probability
  is at most $\eps^{12}$.
  Thus, taking linearity of expectation
  we have that the expected size 
  of $S_0$ is at most $\eps^{12} |S \cup \iota|$ and
  that the expected
  value of $\tw(S_0)$ is at most $\eps^{12}\tw(S \cup \iota)$.
  Applying Markov's inequality to 
  obtain concentration bounds on both quantities and
  then taking a union 
  bound over the probabilities of failure yields the lemma.
\end{proof}

We finally conclude this section with some
additional definitions
that are used in the following sections.
We define a \emph{basic region} of a decomposition
of a to be a region of the dissection that contains exactly
1 points of $L$. The other squares of the decomposition 
are simply called \emph{regions}. 


\section{A Structured Near-Optimal Solution}
\label{sec:structthm}

This section is dedicated to the following proposition.
\begin{prop}
  \label{prop:nearopt}
  Let $L$ be any solution.
  Let $\calD_L$ be a random quadtree dissection 
  of $L$ as per Sec.~\ref{sec:diss}. 
  Suppose Event $\event(L \cup \opt,\tw)$ happens. 
  Then there exists
  a constant $c$ and a solution $S^*$ 
  of cost at most
  $(1+c \cdot \eps) \cost(\opt)  + c\cdot \eps \cost(L)$
  and such that the set of $\irr$ centers
  of $S^*$ is equal to the set of $\irr$ centers of $L$.
\end{prop}

We prove Propositon~\ref{prop:nearopt} by
explicitly constructing $S^*$.
We iteratively modify $\opt$ in four main steps:
\begin{enumerate}
\item Modify $\opt$ by replacing $f_0$ by $\ell_0$ for 
  for each 1-1 isolated pair $(f_0,\ell_0)$ 
  where $\ell_0$ or $f_0$ is a $\irr$ center.
  This yields a near-optimal solution $S_0$ 
  (Lemma~\ref{lem:1-1replacement}).
\item Replace in $\opt$ each $\irr$ center $s$ that is in
  $\iota$ by $\phi(s)$ (as per Section~\ref{sec:diss}). This yields a near-optimal
  solution $S_1$ (Lemma~\ref{lem:non1-1rep}).
\item Apply Theorem~\ref{thm:deletion} (\ie Theorem III.7 in~\cite{CAKM16}) 
to obtain a near-optimal 
  solution $S_2$ that has at most $k -c_2 \eps^9\cdot \widebar{k}$
  where $\widebar{k}$ is the number of facilities of $\opt$ that
  are not 1-1 isolated (Lemma~\ref{lemma:makingroom}).
\item Add the $\irr$ centers of $L$ that are non
  1-1 isolated to $S_2$. This yields a near-optimal 
  solution $S_3$ that has at most
  $k$ centers.
\end{enumerate}

See Section~\ref{subsec:proofofpropstruct} 
for a detailed proof.

\subsection{1-1 Isolated Pairs}
We start from $\opt$ and for each 1-1 isolated
facility $(f_0,\ell_0)$,
$f_0 \in \opt$, $\ell_0 \in L$,
where $\ell_0$ or
$f_0$ is a $\irr$ center, we replace 
$f_0$ by $\ell_0$ in $\opt$. 

This results in a solution $S_0$ whose structural properties
are captured by the following lemma. For any $f_0 \in \opt$
(resp. $\ell_0 \in \opt$), let $\Vglob(f_0)$ (resp. 
$\Vloc(\ell_0)$) be 
the set of clients served by $f_0$ in $\opt$ (resp. 
the set of clients served by $\ell_0$ in $L$).


\begin{lemma}
  \label{lem:1-1replacement}
  Assuming Event $\event(L\cup\opt,\tw)$ happens, 
  there exists a constant $c_0$ such that 
  $\cost(S_0)\le (1+c_0 \cdot \eps) \cost(\opt) + c_0 \cdot 
  \eps \cost(L)$. 
\end{lemma}
\begin{proof}
  Since Event $\event(L\cup\opt,\tw)$ happens,
  we have by Lemma~\ref{lemma:keyprob} 
  that  the total opening cost of the $\irr$ centers plus
  the total service cost induced by the clients served by the 
  $\irr$ centers is 
  bounded by $c_1 \cdot \eps\cdot (\cost(L) + \cost(\opt))$
  for some constant $c_1$. More formally, we can write:
  $$
  \sum_{(f_0,\ell_0):~\text{1-1 
      isolated pair and $\ell_0$ or $f_0$ is a $\irr$ center}} 
  (w(\ell_0)  + \sum_{a \in \Vloc(\ell_0)} \dist(a, \ell_0)) \le 
  c_1 \cdot \eps\cdot (\cost(L) + \cost(\opt))
  $$

  Thus, we need to bound the cost for the clients that are
  in $\Vglob(f_0) - 
  \Vloc(\ell_0)$, for each 1-1 isolated pair $f_0,\ell_0$ where 
  $\ell_0$ or $f_0$ is a $\irr$ center.
  Consider such a pair $f_0,\ell_0$.
  We bound the cost of the clients served by $f_0$ in $\opt$ 
  by the cost of rerouting them toward $\ell_0$. 

  We can thus write for each such client $c$:
  $\dist(c,\ell_0)^2 \le (\dist(c, f_0) + \dist(\ell_0,f_0))^2$.
  Also,
  $\dist(c,\ell_0)^2 \le (1+\eps)^2\dist(c, f_0)^2 + 
  (1+\eps^{-1})^2\dist(\ell_0,f_0))^2$. Note that $\dist(c,f_0)^2$ is
  the cost paid by $c$ in $\opt$.
  Thus we aim at bounding $\dist(\ell_0,f_0)$.

  Applying the triangle inequality, we obtain
  $\dist(\ell_0,f_0)^2 \le (\dist(\ell_0,c_1) + \dist(c_1,f_0))^2$, 
  for any $c_1$ in $\Vglob(f_0) \cap \Vloc(\ell_0)$.
  Hence,
  $$\dist(\ell_0,f_0)^2\le \frac{1}{|\Vglob(f_0) \cap \Vloc(\ell_0)|}
  \sum_{c_1 \in \Vglob(f_0) \cap \Vloc(\ell_0)} 
  (\dist(\ell_0,c_1) + \dist(c_1,f_0))^2.$$
  Now,   
  \begin{align*}
  \sum_{c_1 \in \Vglob(f_0) \cap \Vloc(\ell_0)} (\dist(\ell_0,c_1) + \dist(c_1,f_0))^2
  &\le 3\sum_{c_1 \in \Vglob(f_0) \cap \Vloc(\ell_0)} \dist(c_1,\ell_0)^2 + 
    \dist(c_1,f_0)^2.    
  \end{align*}

  Combining, we obtain that the total cost for the clients
  in $\Vglob(f_0) - \Vloc(\ell_0)$ is at most
  \begin{align*}
    \sum_{c \in \Vglob(f_0) - \Vloc(\ell_0)} \dist(c,\ell_0)^2 &\le
    (1+\eps)^2 \left(\sum_{c \in \Vglob(f_0) - \Vloc(\ell_0)}\dist(c,f_0)^2 \right) + 
     |\Vglob(f_0) - \Vloc(\ell_0)|(1+\eps^{-1})^2\dist(\ell_0,f_0)^2\\
    &\le (1+\eps)^2 \left(\sum_{c \in N(f_0) - N(\ell_0)}\dist(c,f_0)^2 \right)
      \\&+ (1+\eps^{-1})^2 
      \frac{|\Vglob(f_0) - \Vloc(\ell_0)|}{|\Vglob(f_0) \cup \Vloc(\ell_0)|} 
      \cdot 3\sum_{c_1 \in \Vglob(f_0) \cap \Vloc(\ell_0)} \dist(c_1,f_0)^2 + 
          \dist(c_1,\ell_0)^2 \\
  \end{align*}
  The lemma follows from applying the definition of 1-1 isolation:
  $|\Vglob(f_0) - \Vloc(\ell_0)|/|\Vglob(f_0) \cup \Vloc(\ell_0)| \le \eps^3$, and
  summing over all 1-1 isolated pair.
\end{proof}


\subsection{Replacing the Moat Centers of $\opt$}
In this section, we consider the solution $S_0$ and define
a solution $S_1$ whose set of $\irr$ centers is a
subset of the $\irr$ centers of $L$. Namely:
\begin{lemma}
  \label{lem:non1-1rep}
  There exists a constant $c_1$ and a solution $S_1$ such that
  the $\irr$ centers of $S_1$ are a subset of the $\irr$ centers
  of $L$ and $\cost(S_1) \le (1+c_1 \eps) \cost(S_0) +
  c_1\eps\cost(L)$.
\end{lemma}
\begin{proof}
  Note that by Lemma~\ref{lem:1-1replacement},
  all the 1-1 isolated facilities
  of $S_0$ that are $\irr$ centers are also in $L$.
  Thus, we focus on the $\irr$ centers of $S_0$ that
  are not $1-1$-isolated (and so, by definition, in $\iota$).
  
  We replace each center $s \in \iota$ by the center
  $\phi(s) \in L$ (as per Section~\ref{sec:diss}):
  the bound on the cost follows
  immediately from the definition of $\tw$ and
  by combining Lemma~\ref{lemma:keyprob} and
  Lemma~\ref{lem:newreass}.
\end{proof}

\subsection{Making Room for Non-Isolated Moat Facilities}
We now consider the solution $S_1$
described in the previous section
that satisfies the condition of Lemma~\ref{lem:1-1replacement}.
The following lemma is a direct corollary of 
Theorem~\ref{thm:deletion} (from~\cite{CAKM16}).

\begin{lemma}
  \label{lemma:makingroom}
  Given $S_1$ and $L$, there exists a solution $S_2 \subseteq S_1$ 
  and constants $c_2,c_3$ 
  such that 
  \begin{enumerate}
  \item $|S_2| \le k - c_2 \cdot \eps^9 \widebar{k}$,
    where $\widebar{k}$ is the number of facilities that 
    are not 1-1 isolated in $\opt$ (or in $L$ it is the same number), 
    and
  \item $\cost(S_2) \le (1+c_3 \cdot \eps) \cost(\opt) + 
    c_3 \cdot \eps \cdot \cost(L)$.
  \end{enumerate}
\end{lemma}

\subsection{Adding the  Non-Isolated Moat Facilities
and Proof of Proposition~\ref{prop:nearopt}}
\label{subsec:proofofpropstruct}
We now consider the set of centers $S_3$ consisting of 
the centers in $S_2$
and the non-1-1-isolated $\irr$ centers of 
$L$.


\begin{proof}[Proof of Proposition~\ref{prop:nearopt}]
  Combining
  Lemmas~\ref{lem:1-1replacement},~\ref{lem:non1-1rep},
  and~\ref{lemma:makingroom}
  yields the bound on the cost of $S_3$.

  By definition of $S_3$ and applying
  Lemmas~\ref{lem:1-1replacement} and~\ref{lem:non1-1rep},
  we have that the set of $\irr$  centers of $S_3$
  is exactly the set of $\irr$ centers of $L$.

  By Lemma~\ref{lemma:makingroom} (and because
  Event $\event(L \cup \opt,\tw)$ happens), we have that
  the total number of centers in $S_3$ is at most
  $k$.

  
\end{proof}

\section{Proof of Theorem~\ref{thm:main}}
\label{sec:mainthm}
We summarize:
By Proposition~\ref{prop:nearopt}, we have that there exists
a near-optimal solution $S^*$ whose set of $\irr$ centers
is the set of $\irr$ centers of $L$. 
By Proposition~\ref{prop:DP}, we have that, 
\textsc{FindImprovement} identifies a solution $S'$ 
that is $\delta$-close w.r.t. $L$, whose set of
$\irr$ centers is the set of $\irr$ centers
of $L$,
and such that $\cost(L)-\cost(S') \ge (1-\eps) (\cost(L) - 
\cost(\opt_{\delta}))$, where $\opt_{\delta}$ is the minimum cost 
solution $S$ such that $|S - L| + |L - S| \le \delta$.
We now argue  that:
If \textsc{FindImprovement} outputs a solution $S_4$
such that $\cost(L) - \cost(S_4) \le \eps \cost(\opt)/k$ 
then there exists a constant $c^*$ such that 
$\cost(L) \le (1+c^* \cdot \eps)\cost(\opt)$.


Assuming $\cost(L) - \cost(S_4) \le \eps \cost(\opt)/k$ implies 
by Proposition~\ref{prop:DP} that 
$\cost(L) - \cost(\opt_{\delta}) \le 2\eps \cost(\opt)/k$, 
since $\eps < 1/2$.
Now, consider the solution $S^*$ defined
in Section~\ref{sec:structthm}. 

By Theorem 1
in~\cite{FRS16a} 
(see also~\cite{CAKM16} for 
a slightly better dependency in $\eps$
in the unweighted case), if for any pair of sets
$\Delta_1 \subseteq L, \Delta_2 \subseteq S_2$ such that 
$|\Delta_2| \le |\Delta_1| = (d \eps^{-1})^{O(d)}$, we have 
$\cost(L) - \cost(L - \Delta_1 \cup \Delta_2) \le \eps \opt/k$,
then there exists a constant $c_6$ such that 
$\cost(L) \le (1+c_6 \cdot \eps) \cost(S_2)$.
To obtain such a bound we want to apply
Proposition~\ref{prop:DP}, and so we need to show that
any such solution $M = L - \Delta_1 \cup \Delta_2$
is such that its 
$\irr$ centers are the $\irr$ centers of $L$. This follows
immediately from Proposition~\ref{prop:nearopt}: the
$\irr$ centers of $S^*$ are the $\irr$ centers of $L$.


Therefore, we can apply Proposition~\ref{prop:DP} and we have
$\cost(L) \le (1+c^* \cdot \eps) \cost(\opt)$,
for some constant $c^*$.

We now bound the running time of 
Algorithm~\ref{algo:Clustering}.
\begin{lemma}
  \label{lemma:runningtime}
  The running time of Algorithm~\ref{algo:Clustering} is at most 
  $n \cdot k \cdot (\log n)^{(d\eps^{-1})^{O(d)}}$.
\end{lemma}
\begin{proof}
  By Proposition~\ref{prop:DP}, 
  we only need to bound the number of iterations
  of the \texttt{do-while} loop (lines~\ref{line:do} 
  to~\ref{line:while}) of 
  Algorithm~\ref{algo:Clustering}. 
  Let $\cost(S_0)$ denotes the cost of the initial solution.
  The number of iterations of the \texttt{do-while} loop is
  $$\frac{\log(\cost(S_0) / \cost(\opt))}{\log(\frac{1}{1-1/k})}.$$
  Assuming $\cost(S_0) \le O(\opt)$, we have that 
  the total number of iterations is at most
  $O(k)$. To obtain 
  $\cost(S_0) \le  O(\opt)$ it is possible 
  to use the 
  algorithm of Guha et al.~\cite{guha2000clustering}
  which outputs an $O(1)$-approximation 
  in time $n \cdot k \cdot \text{polylog}(n)$, 
  as a preprocessing step (\ie for the computations
  at line~\ref{line:preproc}), without
  increasing the overall running time.
\end{proof}

We conclude by bounding the probability of failure:
By lemma~\ref{lemma:keyprob}, Event $\event$ happens with 
probability at least 1/2. Since the random dissection is repeated 
independently $c \cdot \log(k)$ times, 
the probability of failure of Event $\event$
for a given iteration of the while loop is at most $2^{-c\cdot \log k}
= k^{-c}$. Now, by Lemma~\ref{lemma:runningtime}, the do-while loop 
is repeated
a total of at most $O(k)$ times, thus the probability of failure
is at most $O(k^{-c+1})$.

\section{A Dynamic Program to Find the Best Improvement}
\label{sec:DP}
For a given solution $L$, we define a solution $L'$ to 
be $\delta$-close from $L$ if $|L-L'| + |L'-L| \le \delta$. Let 
$\opt_{\delta}$ denote the cost of the best solution that is 
$\delta$-close from $L$ and whose set of $\irr$ centers
is the set of $\irr$ centers of $L$.
In the following we refer to this solution
by the \emph{best $\delta$-close solution}.

As a preprocessing step, we round the weights of the centers
to the 
closest $(1+\eps)^i \opt/n$, for some integer $i$.
It is easy to see
that this only modify the total value by a factor $(1+\eps)$.

For each region $R$, we define the center of the region $c_R$ to
be the center of the square $R$.
For each point $p$ that is outside of $R$ and at distance at least
$\eps \partial R/\log n$ of $R$, we define the coordinates
of $p$ w.r.t. $c_R$ as follows. Consider the coordinates of the vector 
$\overrightarrow{c_R p}$,
rounded to the closest $(1+\eps/\log n)^i \eps^{14} \partial R/\log n$,
for some integer $i$. Let $\overrightarrow{\widetilde{c_R p}}$ be the
resulting list of coordinates. Let $s$
the point such that the coordinates of the vector $\overrightarrow{c_R s}$
are equal to the coordinates of $\overrightarrow{\widetilde{c_R p}}$.
We define the coordinates of $p$ rounded w.r.t. $R$ to be the coordinates of $s$.

For each region $R$ we also define the grid $G_R$ of $R$ as the $d$-dimensional grid of
size $2\log n/\eps^{14} \times ...\times 2\log n/\eps^{14}$ on $R$. Note that the distance between
two consecutive points of the grid is $\eps^{14} \partial R/2\log n$.
For each point $p$ that is inside of $R$, we define the coordinates of 
$p$ rounded w.r.t. $R$ to be the coordinates of the closest grid point.

The following fact follows from the definition and
recalling that the region sizes
are in the interval $[1,\text{poly}(n)]$.
\begin{fact}
  \label{fact:size1}
  For any region $R$,
  the number of different coordinates rounded w.r.t. to $R$ is at most 
  $O((\log n/\eps^{14})^{2d}$.
\end{fact}

We now describe the dynamic program. Each entry of the table is
defined
by the following parameters:
\begin{itemize}
\item a region $R$,
\item a list of the
  rounded coordinates w.r.t. $R$ of the centers
  of $L - \opt_{\delta}$ and $\opt_{\delta} - L$,
\item a list of the (rounded) weights of the centers of $L - \opt_{\delta}$ and $\opt_{\delta}- L$.
\item a boolean vector of length $\delta$ indicating
  whether $i$th center
  in the above lists is in $L-\opt_{\delta}$ (value 0)or 
  in $\opt_{\delta}-L$ (value 1).
\end{itemize}

The following fact follows from the definition and Fact~\ref{fact:size1}.
\begin{fact}
  The total number of entries that are parameterized by region $R$ is at most
  $(\log n/\eps^{14})^{O(d \cdot \delta)}$.
\end{fact}

We now explain how to fill-up the table. We maintain the
following constraint when we compute a (possibly partial)
solution $L'$: there is no center of $(L' - L) \cup (L-L')$ that
is $\irr$. Under this constraint, we proceeds
as follows, starting with the basic regions which define
the base-case of our DP.
The base-case regions contains only one single
candidate center. Hence, the algorithm proceeds as follows:
it fills up table entries that are parameterized by:
\begin{itemize}
\item any boolean vector of length $\delta$.
\item the rounded coordinates of the unique candidate center
  inside $R$ and a set of
  $\delta-1$ rounded coordinates for the centers of $(L-\opt_{\delta}) \cup (\opt_{\delta}-L)$ outside $R$, or
\item a set of
  $\delta$ rounded coordinates for the centers of $(L-\opt_{\delta}) \cup (\opt_{\delta}-L)$ outside $R$.
\end{itemize}
Additionally, we require that the boolean vector is consistent
with the rounded coordinates: the candidate center inside
$R$ is already in $L$ if and only if
its corresponding boolean entry is 0.

It iterates over all possible rounded coordinates for the at most $\delta$ centers
of $(L-\opt_{\delta}) \cup (\opt_{\delta}-L)$ outside $R$ and for each of possibility, it computes
the cost. Note that this can be done in time $n \cdot \delta$.

We now consider the general case which consists in merging table
entries of child regions.
Fix a table entry parameterized by a region $R$, and
the rounded coordinates of the centers of $(L-\opt_{\delta}) \cup (\opt_{\delta}-L)$. 
We define which 
tables entries of the child regions are compatible given the rounded
coordinates.
For a table entry of a child region $R_1$, with the rounded coordinates of the centers
of $(L-\opt_{\delta}) \cup (\opt_{\delta}-L)$, we require the following for all center $c_0 \in (L-\opt_{\delta}) \cup (\opt_{\delta}-L)$.
Denote by $\tilde{c}_0^1$ the coordinates of $c_0 \in (L-\opt_{\delta}) \cup (\opt_{\delta}-L)$
rounded w.r.t. $c_{R_1}$, namely its values in the table entry
for $R_1$. Let $\tilde{c}_0^R$ denote
its rounded coordinate w.r.t. $R$, namely its values
in the table entry for $R$.
We require:
\begin{itemize}
\item If $\tilde{c}_0^R$ is outside $R$, we say that the table entries are compatible for $c_0$
  if the coordinates of the vector $\overrightarrow{c_R \tilde{c}_0^1}$ and
  $\overrightarrow{c_R \tilde{c}_0^R}$ are all within a $(1\pm\eps/\log n)$ factor.
\item If $c_0$ is inside $R$, we say that the table entries are compatible for $c_0$
  if the point of the grid $G_R$ that is the closest to $\tilde{c}_0^1$ is $\tilde{c}_0^R$.
\item the entries corresponding to $c_0$ in the boolean vectors
  are the same.
\end{itemize}

The following lemma follows immediately from the above facts and the definition.
\begin{lemma}
   \label{lem:dpRT}
  The running time of the dynamic program is $n (\log n/\eps)^{O(d\delta)}$.
\end{lemma}

We now turn to the proof of correctness. For a given region $R$, and a $\delta$-close
solution $S$
we define the table entry of $R$ induced by $S$ to be the table entry parameterized
by $R$ and the coordinates of the centers of $(L-S) \cup (S-L)$ rounded w.r.t. $R$.
\begin{lemma}
  \label{lem:dpcorrect}
  Consider the best $\delta$-close 
  solution $\opt_{\delta}$. For any level $i$ of the quadtree
  dissection, for any region $R$ at level $i$,
  we have that the table entry induced by
  $\opt_{\delta}$ has cost at most
  $ \sum_{c \in R} ((1+\eps/\log n)^i \dist(c, \opt_{\delta}))^2$.
\end{lemma}
\begin{proof}
  Observe that we consider a $\delta$-close 
  solution that has the same set of $\irr$ centers than $L$. Hence, if a client is served by a $\irr$ center in $\opt_{\delta}$, we know exactly the position of this center (as it is also in $L$ and so cannot be removed).
  Thus, for any region $R$, the set of clients in $R$ is served
  by either a center in $R$
  or a center at distance at least $\eps^{13}\partial R/\log n$ from
  the boundary of $R$ or a center of $L$ that is a $\irr$ center.
  
  We now proceed by induction. We consider the base case:
  let $R$ be a region at the maximum level.
  Consider the table entry induced by $\opt_{\delta}$.
  We claim that for each client $c$ in $R$, the cost induced by the solution for this table entry is
  at most $(1+\eps/\log n) \dist(c,\opt_{\delta})$. Indeed, since $\opt_{\delta}$ has the same set of $\irr$ centers
  each client that is served by a center outside of $R$ that
  is at distance at most $\eps^{13} \partial R /\log n$ from the boundary is served by a $\irr$ center of $L$ and so there is no approximation in its service cost. Each client that is served by a center of $\opt_{\delta} - L$ is at distance at least $\eps^{13} \partial R/\log n$
  and so, the error induced by the rounding is at most $\eps\dist(c,\opt_{\delta})/\log n$.
  Finally,
  the cost for the clients in $R$ served by the unique center of $R$ (if there is one) is exact.

  Thus, assume that this holds up to level $i-1$. Consider a region $R$ at level $i$ and the
  table entry induced by $\opt_{\delta}$. The inductive hypothesis implies that for each of the
  table entries of the child regions that are induced by $\opt_{\delta}$, the cost for the clients
  in each subregion is at most $(1+\eps/\log n)^{i-1} \sum_{c \in R} \dist(c,\opt_{\delta})^2$.
  
  By definition, we have that each client of $R$ that is served in $\opt_{\delta}$ by a center
  that is outside $R$ is at distance at least $\eps^{13} \partial R/\log n$ or a $\irr$ center of $L$ (and so the distance is
  known exactly).
  Thus the rounding error incurred for the cost of the clients
  of $R$ served by a center outside $R$ is at most
  $(1+\eps/\log n)^{i} \dist(c, \opt_{\delta})$.

  We now turn to the rounding error introduced 
  for
  the centers that are inside $R$. Let $c$
  be such a center.
  We have that the error introduced is at most the distance
  between two consecutive grid points
  and so at most $\eps^{14} \partial R/2\log n$. Now, observe that
  again because $\opt_{\delta}$ shares the
  same $\irr$ centers
  than $L$,  each client $a$ pf a
  child region $R_1$ that suffers some rounding
  error and that is served by $c$,
  is at distance at least $\eps^{13}\partial R_1/\log n$
  from $c$ and so, combining with the inductive hypothesis
  the error incurred is at most $(1+\eps/\log n)^i\dist(a,c)$.
\end{proof}

Proposition~\ref{prop:DP} follows from combining Lemmas~\ref{lem:dpRT} and~\ref{lem:dpcorrect}.
\begin{prop} 
  \label{prop:DP}
  Let $\eps>0$ be a small enough constant.       
  Let $L$ be a solution and $\calD$ be a decomposition
  and $\calM$ be the $\irr$ centers of $L$.
  The dynamic program FindImprovement output a solution 
  $\opt_{\delta}$ such that 
  $\cost(L) - \cost(\opt_{\delta}) \ge 
  (1-\eps)(\cost(L) - \cost(\opt_{\delta}))$,
  where $\opt_{\delta}$ is a minimum-cost
  solution that is $\delta$-close
  from $L$ and
  whose set of $\irr$ centers is the set of
  $\irr$ centers of $L$.
  Its running time is
  $n \cdot (\log n/\eps^{14})^{O(d \delta)}$.
\end{prop}

\bibliographystyle{abbrv}
\bibliography{facilitylocationptas.bib}
\end{document}